\pdfoutput=1
\documentclass{l4dc2026}

\usepackage{csquotes}
\usepackage{dsfont}
\usepackage{times}
\usepackage{booktabs}

\newtheorem{assumption}[theorem]{Assumption}
\newcommand{\norm}[1]{\left\lVert#1\right\rVert}

\title[Multi-Objective Policy Gradient via Approachability]{An Online Multi-Objective Policy Gradient for Average-Reward Markov Decision Processes under Adversarial Disturbances}

\author{%
 \Name{Rahul Misra} \Email{rmi@es.aau.dk}\\
 \Name{Rafa{\l} Wisniewski} \Email{raf@es.aau.dk}\\
  \Name{Manuela Bujorianu} \Email{lmbu@es.aau.dk}\\
 \addr Department of Electronic Systems, Automation and Control, Aalborg University, Fredrik Bajers Vej 7C, 9220 Aalborg East, Denmark%
}

\begin{document}

\maketitle

\begin{abstract}
We propose a policy gradient algorithm for multiobjective decision-making on continuous state spaces, in which the long-run average of a reward vector must be steered into a target set despite an adversarial disturbance, by scalarizing the reward along the direction prescribed by Blackwell's Approachability Theorem instead of tuning weights by trial and error. Existing approachability-based learning is confined to countable state spaces, because the scalarization is updated at the visits to a recurrent state, which on a continuous state space has probability zero. We update it instead at the returns to a set after a minimum number of steps, show that a geometric drift condition, uniform over the policy class, controls the cycle lengths while the policy is being learned, and replace the parallel learners of the existing algorithm by a single two-timescale actor-critic. We show that the distance to the target set is asymptotically bounded by the per-cycle remainder from Blackwell's condition, and decompose this remainder into the errors of the actor, the policy class, the direction tracking and the return set. We illustrate the algorithm on a heating, ventilation, and humidity control problem with adversarial weather and occupancy.
\end{abstract}


\section{Introduction}

Many control problems are inherently multiobjective with conflicting objectives (a building climate controller must keep the occupants comfortable, keep the humidity within bounds, and minimize energy). The existing literature on multiobjective reinforcement learning (RL) based control can be divided into two categories. The first approach is to scalarize the reward vector in a user-defined manner and to verify the result by numerical simulations \citep{li2020deep,huang2024reinforcement}. It has following drawbacks: Firstly, reward design is a very domain-specific task requiring thorough knowledge of the domain. Secondly, tuning the weights by simulation becomes computationally very expensive when the task requires many samples. Lastly, there are no guarantees, as a better set of weights may exist than the one found by trial and error. The second approach is to treat the additional objectives as constraints of a Constrained Markov Decision Process (CMDP) \citep{altman1999constrained}, for which Bellman's principle of optimality is not always valid \citep{haviv1996constrained} and which is not robust to adversarial disturbances \citep{avrachenkov2016singularly}. A third approach, and the one we follow, is based on Blackwell's Approachability Theorem \citep{blackwell1956analog}. Here we define a target set via the intersection of constraints on multiple objectives. The long-run average of the reward vector is steered into the predefined Target set by adaptively scalarizing the reward along the unit vector from the current average to its projection on the set, and for convex sets Blackwell's condition is necessary and sufficient to \emph{Approach} the set despite any adversarial disturbance \citep{perchet2013approachability}. \citet{shimkin1993guaranteed} and \citet{mannor2004geometric} extended the theorem to stochastic games by updating the scalarization at the successive visits to a recurrent state $\tilde{x}$, and \citet{mannor2004geometric} learn with $J$ RL algorithms in parallel, one per candidate direction. \citet{milman2006approachable} removed the recurrence assumption at the price of a policy computation requiring constants unavailable to a learner, \citet{kalathil2017approachability} gave an RL algorithm for Stackelberg models with an asymptotic guarantee, and \citet[Sec. 10]{jaskiewicz2018zero} state approachability for finite state spaces only. A separate line improves approachability rates via no-regret learning \citep{abernethy2011blackwell}, online convex optimization \citep{shimkin2016online}, pseudonorms \citep{dann2023pseudonorm} and approximation oracles \citep{garber2025approximation}. \citet{misra2026robust} modify the theorem to learn robust correlated equilibria under time-varying disturbances, but all these results are for stateless repeated games. We seek to extend these results to continuous state spaces with finite action space. On a continuous state space, a single state has probability zero under any stationary distribution with a density, so by Kac's Theorem its expected return time is infinite \citep[Thm. 10.2.2]{meyn2012markov} and the cut at $\tilde{x}$ is ill-posed. Player $2$ is not a second learner but models an exogenous disturbance of unknown distribution, here weather and occupancy, so the model is a robust MDP \citep{iyengar2005robust, nilim2005robust} or a game against nature \citep[Sec. 3]{jaskiewicz2018zero}.

\par \emph{Contributions}: Firstly, we extend approachability-based multiobjective RL to Borel state spaces. The recurrent state of \citet{shimkin1993guaranteed} and \citet{mannor2004geometric} is replaced by a return set $C$ visited after a minimum cycle length $L_{\mathrm{min}}$, and condition (V4) of \citet[eq. (15.28)]{meyn2012markov}, imposed uniformly over the policy class, bounds the moments of the cycle lengths along the learning trajectory. Secondly, a single two-timescale actor-critic \citep{borkar2022stochastic} replaces the $J$ parallel learners of \citet{mannor2004geometric}. Thirdly, we show that the distance to the Target set is asymptotically bounded by the per-cycle remainder from Blackwell's condition, almost surely and uniformly over the disturbance. Lastly, we evaluate the method on a heating, ventilation, and humidity control problem against four baselines. 

\section{Problem formulation}\label{sec:formulation}

\emph{Notation}: $\mathcal{P}(A)$ is the space of probability measures on a Borel space $A$, $\langle a,b \rangle$ the inner product, $\norm{\cdot}$ the Euclidean norm, $[a]^+ := \max\{a,0\}$ and $a \wedge b := \min\{a,b\}$. $D(a,T) := \inf_{b \in T} \norm{a - b}$ is the point-to-set distance, $\Pi_T(a) := \arg\min_{b \in T} \norm{a - b}$ the projection onto a closed convex $T$, and $T^{\varepsilon} := \{a : D(a,T) \leq \varepsilon\}$. A boldface $\mathbf{r}$ denotes a reward vector and $r$ a scalar reward. Let $(X,\mathcal B(X))$ be a Borel state space and, for $i\in\{1,2\}$, let $U^i$ be finite action spaces. $\pi_\theta$ represents the joint policy i.e. $\pi_{\theta} = (\pi^1_{\theta},\pi^2)$. We consider the following MDP
\begin{equation*}
    \mathcal M := \left( X, U^1, U^2, P, \mathbf{R}, T, \mu_0 \right),
\end{equation*}
from Player $1$'s perspective, i.e.\ Player $2$'s actions may be arbitrary, where $P$ is a stochastic kernel from $X\times U^1\times U^2$ to $\mathcal P(X)$, $\mathbf{R}: X\times U^1\times U^2 \to\mathbb R^k$ is a bounded, Borel-measurable reward vector of $k$ objectives, $T\subset \mathbb R^k$ is the predefined Target set for Player $1$, and $\mu_0\in\mathcal P(X)$ is the initial distribution.  Player $1$ uses the Boltzmann policy
\begin{equation}\label{eq:policy}
    \pi^1_{\theta}(u\mid x) = \frac{\exp\big(\langle\theta,\phi(x,u)\rangle/\zeta\big)} {\sum_{u'\in U^1}\exp\big(\langle\theta,\phi(x,u')\rangle/\zeta\big)},
\end{equation}
where $\theta$ is the policy parameter in a set $\Theta\subset\mathbb R^{d}$, $\phi(x,u)\in\mathbb R^{d}$ is a feature vector of the state-action pair, and the exploration parameter $\zeta>0$ controls the exploration (larger $\zeta$, makes the policy uniform distribution). The set $\{\pi^1_{\theta} : \theta\in\Theta\}$ of these policies is the \emph{policy class}. The initial state is $x_0\sim\mu_0$. At time $t$ both players observe $x_t\in X$, Player $1$ picks $u^1_t\in U^1$ and Player $2$ picks $u^2_t\in U^2$. The reward vector $\mathbf r_t := \mathbf{R}(x_t,u^1_t,u^2_t)$ is observed and $x_{t+1}\sim P(\cdot\mid x_t,u^1_t,u^2_t)$. $P$ and $\mathbf{R}$ are unknown, and only the realized trajectory is observed. $T$ is non-empty, closed and convex, which makes $\Pi_T$ single-valued and Blackwell's condition necessary and sufficient \citep{blackwell1956analog}. A policy of Player $2$ is a stochastic kernel $\pi^2$ from $X$ to $\mathcal P(U^2)$ and $\Pi_2$ denotes the set of such policies, whereas a \emph{strategy} of Player $2$ may depend on the entire history. A joint policy $\pi=(\pi^1_\theta,\pi^2)$ induces a Markov chain on $X$ with kernel,
\begin{equation*}
    P^{\pi}(B\mid x)=\int_{U^1}\int_{U^2}P(B\mid x,u^1,u^2)\,\pi^2(du^2\mid x)\,\pi^1(du^1\mid x),\qquad B\in\mathcal B(X).
\end{equation*}
The time-averaged reward vector is 
\begin{equation*}
    \bar{\mathbf r}_t := \frac{1}{t}\sum^{t-1}_{s = 0} \mathbf{r}_s, \quad \text{where } t \geq 1.
\end{equation*}
\begin{definition}\label{defn:Approachability}
Given the initial distribution $\mu_0$, a Player $1$ strategy \emph{Approaches} the target set $T\subset\mathbb R^k$ if
\begin{equation}\label{eq:Optimality}
    \lim_{t \to \infty} D(\bar{\mathbf{r}}_t, T) = 0 \quad \text{almost surely, for every strategy of Player } 2.
\end{equation}
If \eqref{eq:Optimality} holds with $T$ replaced by $T^{\varepsilon}$ we say that the strategy \emph{$\varepsilon$-Approaches} $T$.
\end{definition}
Since \eqref{eq:Optimality} must hold for arbitrary actions of Player $2$, we can focus on developing an algorithm for Player $1$ that is robust to bounded adversarial disturbances.

\section{When can the scalarization be updated?}\label{sec:regeneration}

A long-run average problem has no episodes, so the outer loop needs time instants at which to act. In the finite setting these come for free. A finite stochastic game is \emph{irreducible} if every pair of stationary strategies induces a Markov chain with a single recurrent class \citep[Def. 5.1.1]{filar2012competitive}. Then the long-run average reward does not depend on the initial state \citep[Thm. 5.1.3]{filar2012competitive}, and one may cut the trajectory at the successive visits to a fixed state $\tilde{x}$,
\begin{equation}\label{eq:finite_cut}
    \tau_n = \min \{ t > \tau_{n-1} \mid x_t = \tilde{x} \},
\end{equation}
and treat the stretches between visits as the stages of Blackwell's game \citep{shimkin1993guaranteed,mannor2004geometric}. On a continuous state space \eqref{eq:finite_cut} is ill-posed, as a state of probability zero has infinite expected return time due to Kac's Theorem \citep[Thm. 10.2.2]{meyn2012markov}. We therefore cut at the returns to a \emph{set} $C$, but only after a minimum number $L_{\mathrm{min}}\geq 1$ of steps,
\begin{equation}\label{eq:regen}
    \tau_0 := 0, \qquad \tau_n := \min \{ t \geq \tau_{n-1} + L_{\mathrm{min}} \mid x_t \in C \}, \quad n \geq 1,
\end{equation}
with cycle lengths $L_n := \tau_n - \tau_{n-1} \geq L_{\mathrm{min}}$. The minimum length keeps the cycles from collapsing to a single step in $C$. Let $V: X \to [1,\infty)$ be a measurable Lyapunov function. The expected one-step change of $V$ from the state $x$ under $\pi$ is 
\begin{equation*}
    \Delta V(x) :=  \mathbb{E}[V(x_{t+1}) \mid x_t = x] - V(x) = \int_X V(y) P^{\pi}(dy \mid x) - V(x).    
\end{equation*}
The following assumption is the continuous-state version of the irreducibility assumption for finite Markov chains.
\begin{assumption}[Uniform geometric drift towards $C$]\label{assump:drift}
There exist a set $C \subset X$, a function $V : X \to [1,\infty)$ that is bounded on $C$, and constants $\beta \in (0,1)$, $b < \infty$ such that, for every $\theta \in \Theta$ and every $\pi^2\in\Pi_2$, the kernel $P^{\pi}$ with $\pi = (\pi^1_{\theta}, \pi^2)$ satisfies
\begin{equation}\label{eq:drift}
    \Delta V(x) \leq -\beta V(x) + b \mathds{1}_C(x), \qquad x \in X.
\end{equation}
and $C$ is \textbf{small}, i.e. $\exists m \in \mathbb{N}$, $\delta > 0$ and a probability measure $\nu$ with $P^m(\,\cdot \mid x) \geq \delta \nu(\cdot), \forall x \in C$, where $P^m$ is the $m$-step kernel of $P^{\pi}$.
\end{assumption}
Condition \eqref{eq:drift} is from \citet[eq. (15.28)]{meyn2012markov} (requires \emph{smallness} of $C$), imposed for every player 1 policy parameter $\theta$, and player 2 policy $\pi$. It can be rewritten as $\mathbb{E}[V(x_{t+1}) \mid x_t = x] \leq (1-\beta) V(x) + b \mathds{1}_C(x)$. Thus, outside $C$, the expected value of $V$ reduces by a factor $1-\beta$ at every step, so the state drifts towards $C$, while inside $C$ it may grow by at most $b$. Thus, we know that the state will eventually return to $C$, but how long will it take? We answer that in the following Lemma, which bounds the first and second conditional moments of the random cycle lengths.
\begin{lemma}[Cycle lengths]\label{lem:cycles}
Let Assumption \ref{assump:drift} hold and let $\mathcal{F}_t$ be the history up to time $t$. There are constants $M_1, M_2 < \infty$, depending only on $\beta$, $b$, $L_{\mathrm{min}}$, $\sup_C V$ and $V(x_0)$, such that for every action by Player $2$,
\begin{equation}\label{eq:cyclemoments}
    \mathbb{E}\big[L_n \mid \mathcal{F}_{\tau_{n-1}}\big] \leq M_1, \qquad \mathbb{E}\big[L_n^2 \mid \mathcal{F}_{\tau_{n-1}}\big] \leq M_2, \qquad n \geq 1.
\end{equation}
\end{lemma}

\begin{proof}
Let $\gamma := 1-\beta \in (0,1)$ and define $s_n := \tau_{n-1} + L_{\mathrm{min}}$. At time $s_n$, the minimum length has been achieved, and the cycle ends as soon as the state is in $C$. Define $\sigma_n \geq 0$ the number of steps after $s_n$ until the cycle ends. Thus, $\tau_n = s_n + \sigma_n$ and $L_n = L_{\mathrm{min}} + \sigma_n$. Since \eqref{eq:drift} holds for every $\theta\in\Theta$ and every policy of Player $2$, it holds at every time along the trajectory, for every $\theta$, and for every Player $2$ action,
\begin{equation}\label{eq:onestep}
    \mathbb{E}\big[V(x_{t+1}) \mid \mathcal{F}_t\big] \leq \gamma V(x_t) + b \mathds{1}_C(x_t), \qquad t \geq 0.
\end{equation}

Applying \eqref{eq:onestep} $L_{\mathrm{min}}$ times from $\tau_{n-1}$ and summing the geometric series ($\sum^{L_{\mathrm{min}}-1}_{i=0}\gamma^i$),
\begin{equation}\label{eq:Vbar}
    \mathbb{E}\big[V(x_{s_n}) \mid \mathcal{F}_{\tau_{n-1}}\big] \leq \gamma^{L_{\mathrm{min}}} V(x_{\tau_{n-1}}) + \frac{b}{1-\gamma} \leq \bar{V} := \max\Big\{\sup_C V,\, V(x_0)\Big\} + \frac{b}{\beta},
\end{equation}
because for $n \geq 2$, the previous cycle must have ended in $C$ or $x_{\tau_{n-1}} \in C$. 
Now, we consider the case when the state $x_{s_n+j}$ is still outside $C$ or $j < \sigma_n$. Now, \eqref{eq:onestep} becomes 
\begin{equation*}
    \mathbb{E}[V(x_{s_n+j+1}) \mid \mathcal{F}_{s_n+j}] \leq \gamma V(x_{s_n+j}).
\end{equation*} 
Hence $Y_j := \gamma^{-(j\wedge\sigma_n)} V(x_{s_n+j\wedge\sigma_n})$ is a nonnegative supermartingale and thus, $\mathbb{E}[Y_j \mid \mathcal{F}_{s_n}] \leq Y_0 = V(x_{s_n})$. Since $V \geq 1$ (by Assumption \ref{assump:drift}), letting $j \to \infty$ gives
\begin{equation}\label{eq:geometric}
    \mathbb{E}\big[\gamma^{-\sigma_n} \mid \mathcal{F}_{s_n}\big] \leq V(x_{s_n}).
\end{equation}
Thus, the state returns to $C$ almost surely and the return time has a geometrically decaying tail, $\mathbb{P}(\sigma_n > j \mid \mathcal{F}_{s_n}) \leq \gamma^{j} V(x_{s_n})$. Let $\kappa := \ln(1/\gamma) > 0$. Since $\gamma^{-\sigma_n} = e^{\kappa\sigma_n}$ is greater than $\kappa\sigma_n$ and $(\kappa\sigma_n)^2/2$, we can bound the first and second moments using \eqref{eq:Vbar} and \eqref{eq:geometric} as $\mathbb{E}[\sigma_n \mid \mathcal{F}_{\tau_{n-1}}] \leq \bar{V}/\kappa$ and $\mathbb{E}[\sigma_n^2 \mid \mathcal{F}_{\tau_{n-1}}] \leq 2\bar{V}/\kappa^2$. As $L_n = L_{\mathrm{min}} + \sigma_n$ and $L_n^2 \leq 2L_{\mathrm{min}}^2 + 2\sigma_n^2$, \eqref{eq:cyclemoments} holds with $M_1 := L_{\mathrm{min}} + \bar{V}/\kappa$ and $M_2 := 2L_{\mathrm{min}}^2 + 4\bar{V}/\kappa^2$. Since only the one-step inequality \eqref{eq:onestep} is used at each time, the bounds hold, as $\theta$ changes only after every cycle is completed.
\end{proof}

\section{Algorithm}\label{sec:algorithm}

As we consider a vector of rewards, standard RL methods are not directly applicable, and we design an algorithm with two loops. The outer loop adaptively scalarizes the reward vector at the cut times \eqref{eq:regen}, and the inner loop is an average-reward actor-critic acting on the scalarized reward. For $s\notin T$ the unit vector pointing from $s$ towards the Target set is
\begin{equation}\label{eq:Unit_vector}
    \lambda_s = \frac{\Pi_T(s) - s}{\norm{\Pi_T(s) - s}},
\end{equation}
which is well defined since $T$ is closed and convex, and $\lambda_s := \mathbf{0}$ for $s\in T$. For a unit vector $\lambda$, let
\begin{equation*}
    v(\theta,\lambda):=\inf_{x\in X}\,\inf_{\pi^2\in\Pi_2}\,\liminf_{t\to\infty}\,\big\langle\lambda,\,\mathbb{E}^{x}_{\pi_\theta}\bar{\mathbf r}_t\big\rangle, \qquad v^{\star}(\lambda):=\sup_{\theta\in\Theta}v(\theta,\lambda),
\end{equation*}
be the long-run average of the reward scalarized along $\lambda$ that the policy $\pi^1_{\theta}$ guarantees from every initial state and against every policy of Player $2$, and the best such guarantee, where $\mathbb{E}^{x}_{\pi_\theta}$ is the expectation under the joint policy $\pi_\theta$ when $x_0=x$. \emph{Blackwell's condition} \citep{blackwell1956analog,shimkin1993guaranteed} requires that, wherever the current average $s\notin T$ is, Player $1$ can guarantee a long-run average reward on the side of $T$ of the hyperplane through $\Pi_T(s)$ with normal $\lambda_s$, i.e.
\begin{equation}\label{eq:blackwell}
    \langle\lambda_s,\Pi_T(s)\rangle-v^{\star}(\lambda_s) \leq \varepsilon_{\mathrm{approx}}\quad\text{for every } s\notin T,
\end{equation}
where $\varepsilon_{\mathrm{approx}}\geq 0$ is the error for the considered policy class (Boltzmann policy with features $\phi$ in this case). The condition also gives the update rule for Player $1$. When the average is at $s$, maximize the reward scalarized along $\lambda_s$, which is a scalar RL problem. The critic $\hat V(x) = \langle\rho,\phi(x)\rangle$, with parameter $\rho$ and state features $\phi(x)$, and the estimate $\hat g$ of the average scalarized reward (the gain) are updated by average-reward temporal differences at every step, and the actor gradient is accumulated over the cycle and applied at its end, so that each cycle is played by a single stationary policy. The resulting procedure is Algorithm \ref{alg:mopg}, where $\mathrm{Proj}_{\Theta}$ is the Euclidean projection onto $\Theta$. To ensure that the Blackwell condition is enforced, the update is two-timescale. We set the step sizes of the actor ($\alpha_n$), the critic ($\beta_t$, $\beta_{g,t}$) and the direction ($\eta_n$) as
\begin{equation*}
    \alpha_n=\alpha_0(1+n)^{-c_{\alpha}},\qquad \beta_t=\beta_0(1+t)^{-c_{\alpha}},\qquad \eta_n=\eta_0(1+n)^{-c_{\lambda}},\qquad c_{\alpha}<c_{\lambda}<1,
\end{equation*}
and $\beta_{g,t}$ proportional to $\beta_t$, we have $\eta_n/\alpha_n\to 0$, so $\theta$ evolves on the fast and $\lambda$ on the slow timescale \citep{borkar2022stochastic}, while $c_{\lambda}<1$ updates $\lambda$ as per \eqref{eq:blackwell} after a cycle is complete. 

\begin{algorithm2e}[t]
\caption{Two-timescale multiobjective policy gradient for \emph{Approachability}} \label{alg:mopg}
\KwIn{Target set $T$, features $\phi$, return set $C$, minimum cycle length $L_{\mathrm{min}}$, step sizes $\alpha_0,\beta_0,\beta_{g,0},\eta_0$ with exponents $c_{\alpha}<c_{\lambda}<1$, exploration parameter $\zeta$, parameter set $\Theta$}
$\theta\gets\mathbf 0$, $\rho\gets\mathbf 0$, $\hat g\gets 0$, $\bar{\mathbf r}\gets\mathbf 0$, $t\gets 0$, $\lambda\gets\mathbf 1/\sqrt k$ \newline
\For{$n = 1,2,\dots$}{
  $\Gamma\gets\mathbf 0$, $\ell\gets 0$ \tcp*{the policy $\pi^1_{\theta}$ is frozen for the cycle}
  \Repeat{$\ell\geq L_{\mathrm{min}}$ \textbf{\upshape and} $x_{t}\in C$ \tcp*[f]{cut \eqref{eq:regen}}}{
    $u^1_t\sim\pi^1_{\theta}(\cdot\mid x_t)$ as in \eqref{eq:policy}, $u^2_t$ chosen by Player $2$, observe $\mathbf r_t$ and $x_{t+1}$
    $\delta_t\gets\langle\lambda,\mathbf r_t\rangle-\hat g+\langle\rho,\phi(x_{t+1})\rangle-\langle\rho,\phi(x_t)\rangle$, \quad
    $\hat g\gets\hat g+\beta_{g,t}\delta_t$, \quad  $\rho\gets\rho+\beta_t\delta_t\phi(x_t)$ \tcp*{critic, update every step}
    $\Gamma\gets\Gamma+\delta_t\nabla_{\theta}\log\pi^1_{\theta}(u^1_t\mid x_t)$, \quad $\bar{\mathbf r}\gets\bar{\mathbf r}+\tfrac{1}{t+1}(\mathbf r_t-\bar{\mathbf r})$
    $t\gets t+1$, $\ell\gets\ell+1$
  }
  $\theta\gets\mathrm{Proj}_{\Theta}\big(\theta+\alpha_n \Gamma\big)$ \tcp*{actor, update once per cycle}
  $\lambda^{\star}\gets\lambda_{\bar{\mathbf r}}$ as in \eqref{eq:Unit_vector}
  $\lambda\gets\big((1-\eta_n)\lambda+\eta_n\lambda^{\star}\big)/\norm{(1-\eta_n)\lambda+\eta_n\lambda^{\star}}$ \tcp*{slow timescale}
}
\end{algorithm2e}

\section{Convergence analysis}\label{sec:convergence}

We first bound the distance to $T$ by the cycles remainder from Blackwell's condition (Theorem \ref{thm:main}), and then bound the remainder by the learning errors (Proposition \ref{prop:remainder} and Corollary \ref{cor:asymptotic}).

\begin{assumption}[Regularity]\label{assump:regularity}
Let the reward vector be upper bounded i.e. $\norm{\mathbf{R}(x,u^1,u^2)}\leq R_{\mathrm{max}}$ for all arguments. $T$ is non-empty, closed, convex and intersects with $\{y : \norm{y}\leq R_{\mathrm{max}}\}$. $\Theta$ is compact and convex. $\phi$ is bounded and measurable and $\zeta>0$.
\end{assumption}

After the $n$-th cycle, let $d_n := D(\bar{\mathbf r}_{\tau_n},T)$, $p_n := \Pi_T(\bar{\mathbf r}_{\tau_n})$ and $\lambda^{\star}_n := \lambda_{\bar{\mathbf r}_{\tau_n}}$, the exact Blackwell direction \eqref{eq:Unit_vector}, and let $G_n := \sum_{t=(\tau_{n-1})}^{(\tau_n)-1}\mathbf r_t$ be the return vector collected in cycle $n$. Thus $\lambda^{\star}_n$ is the vector $\lambda^{\star}$ computed by Algorithm \ref{alg:mopg} at the end of cycle $n$. We use the convention $\bar{\mathbf r}_0:=\mathbf 0$, and by \eqref{eq:Unit_vector} $\lambda^{\star}_n$ is a unit vector if $\bar{\mathbf r}_{\tau_n}\notin T$ and $\lambda^{\star}_n=\mathbf 0$ otherwise, in which case $d_n=0$.

\begin{definition}[Blackwell remainder]\label{defn:remainder}
The remainder of cycle $n$ is
\begin{equation}\label{eq:remainder}
    \varepsilon_n := \frac{\Big[-\mathbb{E}\big[\langle\lambda^{\star}_{n-1},\,G_n - L_n p_{n-1}\rangle \mid \mathcal{F}_{\tau_{n-1}}\big]\Big]^+}{\mathbb{E}\big[L_n \mid \mathcal{F}_{\tau_{n-1}}\big]} .
\end{equation}
\end{definition}
The remainder $\varepsilon_n$ is the expected amount per step by which cycle $n$ misses Blackwell's condition \eqref{eq:blackwell} at $s=\bar{\mathbf r}_{\tau_{n-1}}$, at the start of a new cycle. The following theorem guarantees that the square of the distance to the Target set $T$, which can grow over the cycles due to remainders \eqref{eq:remainder} is at most $\varepsilon$. 
\begin{theorem}[$\varepsilon$-Approachability]\label{thm:main}
Let Assumptions \ref{assump:drift} and \ref{assump:regularity} hold. For every strategy of Player $2$, the running average generated by Algorithm \ref{alg:mopg} satisfies, almost surely,
\begin{equation}\label{eq:guarantee}
    \limsup_{t\to\infty} D(\bar{\mathbf r}_t, T)\leq\limsup_{n\to\infty}\varepsilon_n .
\end{equation}
In particular, the algorithm $\varepsilon$-\emph{Approaches} $T$ whenever $\varepsilon_n\leq\varepsilon$ for all large $n$.
\end{theorem}

\begin{proof}
The proof strategy is to enlarge $T$ by a margin $\varepsilon>0$. Once $\varepsilon_n\leq\varepsilon$, the enlarged set satisfies Blackwell's condition exactly, and the remainder vanishes. Fix $\varepsilon>0$ and let $m_n:=\mathbb{E}[L_n\mid\mathcal F_{\tau_{n-1}}]$, so that $L_{\mathrm{min}}\leq m_n\leq M_1$ by Lemma \ref{lem:cycles}. By Assumption \ref{assump:regularity}, $\norm{\bar{\mathbf r}_t}\leq R_{\mathrm{max}}$ and $T$ contains a point of norm at most $R_{\mathrm{max}}$, hence $d_n = \norm{\bar{r}_{\tau_n}-y_0} \leq \norm{\bar{r}_{\tau_n}} + \norm{y_0} \leq 2R_{\mathrm{max}}$ and $\norm{p_n} = \norm{\bar{r}_{\tau_n}} + \norm{\bar{r}_{\tau_n} - p_n}\leq R_{\mathrm{max}} + 2R_{\mathrm{max}} = 3R_{\mathrm{max}}$.

Firstly, we derive an inequality for one cycle. Let $p^{\varepsilon}_n:=p_n-\min\{d_n,\varepsilon\}\lambda^{\star}_n$ be the point on the segment from $p_n$ to $\bar{\mathbf r}_{\tau_n}$ at distance $\min\{d_n,\varepsilon\}$ from $p_n$. Then $\norm{p^{\varepsilon}_n}\leq 3R_{\mathrm{max}}$ and, by \eqref{eq:Unit_vector}, $\bar{\mathbf r}_{\tau_n}-p^{\varepsilon}_n=-(d_n-\varepsilon)^+\lambda^{\star}_n$. Let $D_n:=\tau_n(d_n-\varepsilon)^+$, so that $D_0=0$, and $w_n:=G_n-L_np^{\varepsilon}_{n-1}$, so that $\norm{w_n}\leq 4R_{\mathrm{max}}L_n$. Since $\tau_n\bar{\mathbf r}_{\tau_n}=\tau_{n-1}\bar{\mathbf r}_{\tau_{n-1}}+G_n$, we have $\tau_n(\bar{\mathbf r}_{\tau_n}-p^{\varepsilon}_{n-1})=-D_{n-1}\lambda^{\star}_{n-1}+w_n$. As $p^{\varepsilon}_{n-1}$ is within $\varepsilon$ of $p_{n-1}\in T$, the triangle inequality gives $d_n\leq\norm{\bar{\mathbf r}_{\tau_n}-p^{\varepsilon}_{n-1}}+\varepsilon$, so $D_n$ is at most the norm of this vector, and expanding the square gives
\begin{equation}\label{eq:onecycle}
    D_n^2 \leq D_{n-1}^2 - 2D_{n-1}\big\langle\lambda^{\star}_{n-1},\,w_n\big\rangle + 16R_{\mathrm{max}}^2L_n^2 .
\end{equation}

Secondly, we show that all the terms in the remainder vanish. If $D_{n-1}>0$, then $d_{n-1}>\varepsilon$, $p^{\varepsilon}_{n-1}=p_{n-1}-\varepsilon\lambda^{\star}_{n-1}$ and $\langle\lambda^{\star}_{n-1},w_n\rangle=\langle\lambda^{\star}_{n-1},G_n-L_np_{n-1}\rangle+\varepsilon L_n$, whose conditional mean is at least $(\varepsilon-\varepsilon_n)m_n$ by Definition \ref{defn:remainder}. Hence, for every $n$,
\begin{equation*}
    -D_{n-1}\,\mathbb{E}\big[\langle\lambda^{\star}_{n-1},w_n\rangle\mid\mathcal F_{\tau_{n-1}}\big]\leq D_{n-1}(\varepsilon_n-\varepsilon)^+m_n ,
\end{equation*}
which is zero whenever $\varepsilon_n\leq\varepsilon$. Let $\xi_k:=-\langle\lambda^{\star}_{k-1},w_k\rangle+\mathbb{E}[\langle\lambda^{\star}_{k-1},w_k\rangle\mid\mathcal F_{\tau_{k-1}}]$, so that $\mathbb{E}[\xi_k\mid\mathcal F_{\tau_{k-1}}]=0$. Summing \eqref{eq:onecycle} from $D_0=0$ and dividing by $\tau_n^2$,
\begin{equation}\label{eq:sum}
    \big((d_n-\varepsilon)^+\big)^2 \leq \frac{2}{\tau_n^2}\sum_{k=1}^{n}\Big( D_{k-1}\big((\varepsilon_k-\varepsilon)^+m_k + \xi_k\big) + 8R_{\mathrm{max}}^2L_k^2 \Big).
\end{equation}

Lastly, we show that, after division by $\tau_n^2$, the three terms in the sum of \eqref{eq:sum} vanish on the event $\{\limsup_n\varepsilon_n<\varepsilon\}$. For the last term, taking expectations in \eqref{eq:cyclemoments} gives $\mathbb{E}[L_n^2]\leq M_2$, so Markov's inequality applied to $L_n^2$ gives $\mathbb P(L_n>\epsilon n)=\mathbb P(L_n^2>\epsilon^2n^2)\leq M_2/(\epsilon n)^2$ for every $\epsilon>0$, which is summable in $n$. Hence $L_n/n\to 0$ almost surely by the Borel--Cantelli lemma, and then also $\max_{k\leq n}L_k/n\to 0$. Since $\sum_{k\leq n}L_k^2\leq\tau_n\max_{k\leq n}L_k$ and $\tau_n\geq nL_{\mathrm{min}}$, the last term is at most $16R_{\mathrm{max}}^2\max_{k\leq n}L_k/(nL_{\mathrm{min}})$ and therefore vanishes. For the middle term, let $N_n:=2\sum_{k=1}^{n}D_{k-1}\xi_k$. Then $\mathbb{E}[\xi_k^2\mid\mathcal F_{\tau_{k-1}}]\leq 16R_{\mathrm{max}}^2M_2$ because $|\langle\lambda^{\star}_{k-1},w_k\rangle|\leq 4R_{\mathrm{max}}L_k$, and $D_{k-1}\leq 2R_{\mathrm{max}}\tau_{k-1}$ with $\mathbb{E}[\tau_{k-1}^2]\leq k^2M_2$, because $\tau_{k-1}=\sum_{j<k}L_j$ and, by Minkowski's inequality and Lemma \ref{lem:cycles}, $(\mathbb{E}[\tau_{k-1}^2])^{1/2}\leq\sum_{j<k}(\mathbb{E}[L_j^2])^{1/2}\leq(k-1)\sqrt{M_2}$. Hence $\mathbb{E}[(D_{k-1}\xi_k)^2]\leq 64R_{\mathrm{max}}^4M_2^2k^2$, and the martingale $\sum_{k\leq n}D_{k-1}\xi_k/k^2$ has $\sum_k\mathbb{E}[(D_{k-1}\xi_k/k^2)^2]<\infty$. It is therefore bounded in $L^2$ and converges almost surely. Kronecker's lemma states that if $\sum_kx_k/b_k$ converges and $b_k\uparrow\infty$, then $\frac{1}{b_n}\sum_{k\leq n}x_k\to 0$. With $x_k=D_{k-1}\xi_k$ and $b_k=k^2$ it gives $N_n/n^2\to 0$, and therefore $|N_n|/\tau_n^2\leq|N_n|/(nL_{\mathrm{min}})^2\to 0$. For the first term, on the event $\{\limsup_n\varepsilon_n<\varepsilon\}$ there is a (random) $k_0$ such that $\varepsilon_k\leq\varepsilon$ for all $k>k_0$, so only the first $k_0$ terms of the sum are nonzero. Their sum does not depend on $n$, while $\tau_n\to\infty$, so the first term vanishes.

Hence $(d_n-\varepsilon)^+\to 0$, i.e.\ $\limsup_nd_n\leq\varepsilon$, almost surely on $\{\limsup_n\varepsilon_n<\varepsilon\}$. Applying this to every rational $\varepsilon>0$ gives $\limsup_nd_n\leq\limsup_n\varepsilon_n$ almost surely. Lastly, we consider the times between two cycles. For $\tau_{n-1}\leq t\leq\tau_n$, $t\,\bar{\mathbf r}_t=\tau_{n-1}\bar{\mathbf r}_{\tau_{n-1}}+\sum_{s=\tau_{n-1}}^{t-1}\mathbf r_s$, so $\norm{\bar{\mathbf r}_t-\bar{\mathbf r}_{\tau_{n-1}}}\leq 2R_{\mathrm{max}}L_n/\tau_{n-1}$, which vanishes because $L_n/n\to 0$. As $D(\cdot,T)$ is $1$-Lipschitz, $\limsup_tD(\bar{\mathbf r}_t,T)\leq\limsup_nd_n$, which gives \eqref{eq:guarantee}. The strategy of Player $2$ entered only through Lemma \ref{lem:cycles} and Definition \ref{defn:remainder}, which hold for every such strategy.
\end{proof}
The algorithm affects Theorem \ref{thm:main} via the remainder $\varepsilon_n$ defined in \eqref{eq:remainder}. The following assumption relates the long-run average $v(\theta,\lambda)$ guaranteed by a policy to the return of a single cycle. This is needed to decompose $\varepsilon_n$ into constituent terms such as the limitation of policy class $\varepsilon_{\mathrm{approx}}$.  
\begin{assumption}[Relative value function]\label{assump:relvalue}
There is $c_h<\infty$ such that, for every $\theta\in\Theta$ and unit vector $\lambda$, there is a function $h_{\theta,\lambda}$ with $|h_{\theta,\lambda}|\leq c_hV$ and, for all $x\in X$,
\begin{equation}\label{eq:acoi}
    \int_{U^1}\min_{u^2\in U^2}\Big[\big\langle\lambda,\mathbf{R}(x,u^1,u^2)\big\rangle+(Ph_{\theta,\lambda})(x,u^1,u^2)\Big]\,\pi^1_{\theta}(du^1\mid x)\geq v(\theta,\lambda)+h_{\theta,\lambda}(x),
\end{equation}
where $(Ph)(x,u^1,u^2):=\int_X h(y)\,P(dy\mid x,u^1,u^2)$.
\end{assumption}
Inequality \eqref{eq:acoi} is the average-reward optimality inequality of the problem Player $2$ faces when Player $1$'s policy is fixed, as in the game against nature of \citet[Sec. 3]{jaskiewicz2018zero}. The difference $h(x)-h(y)$ is the relative advantage of starting from state $x$ rather than $y$. 
The following Proposition decomposes the Blackwell remainder $\varepsilon_n$ \eqref{eq:remainder} into error terms. 
\begin{proposition}[Decomposition of the remainder]\label{prop:remainder}
Let Assumptions \ref{assump:drift}, \ref{assump:regularity} and \ref{assump:relvalue} hold, and let $\theta_{n-1}$ and $\lambda_{n-1}$ be the actor parameter and the direction used in cycle $n$. Then, for every $n\geq 2$,
\begin{equation}\label{eq:decomp}
    \varepsilon_n\leq\Delta_n+\varepsilon_{\mathrm{approx}}+2R_{\mathrm{max}}\norm{\lambda_{n-1}-\lambda^{\star}_{n-1}}+\frac{H_C}{L_{\mathrm{min}}},
\end{equation}
where $\Delta_n:=v^{\star}(\lambda_{n-1})-v(\theta_{n-1},\lambda_{n-1})$ is the suboptimality of the actor for the direction it optimizes, $\varepsilon_{\mathrm{approx}}$ is as in \eqref{eq:blackwell}, and $H_C:=\sup_{\theta,\lambda}\big(\sup_C h_{\theta,\lambda}-\inf_C h_{\theta,\lambda}\big)$.
\end{proposition}
\begin{proof}
Here we seek to use the relative value function $h$ in \eqref{eq:acoi} to find a bound on the return of a single cycle, which is at most $H_C$ because every cycle starts and ends in $C$. Fix $n\geq 2$ and write $\theta:=\theta_{n-1}$, $\lambda:=\lambda_{n-1}$, $\lambda^{\star}:=\lambda^{\star}_{n-1}$, $p:=p_{n-1}$ and $m_n:=\mathbb{E}[L_n\mid\mathcal F_{\tau_{n-1}}]$. If $\bar{\mathbf r}_{\tau_{n-1}}\in T$, then $\lambda^{\star}=\mathbf 0$ and $\varepsilon_n=0$, so \eqref{eq:decomp} holds. Otherwise $\lambda^{\star}$ is a unit vector, and we let $h:=h_{\theta,\lambda^{\star}}$ and $v:=v(\theta,\lambda^{\star})$.

Firstly, we bound the reward of one step. During the cycle the policy is frozen, so $u^1_t\sim\pi^1_{\theta}(\cdot\mid x_t)$. Whatever Player $2$ does at time $t$, the conditional expectation of $\langle\lambda^{\star},\mathbf r_t\rangle+h(x_{t+1})$ given $\mathcal F_t$ and $u^1_t$ is at least the minimum over $u^2$ in \eqref{eq:acoi}. Integrating over $u^1_t$ and using \eqref{eq:acoi},
\begin{equation}\label{eq:submart}
    \mathbb{E}\big[\langle\lambda^{\star},\mathbf r_t\rangle+h(x_{t+1})\mid\mathcal F_t\big] \geq v+h(x_t), \qquad \tau_{n-1}\leq t<\tau_n .
\end{equation}
Thus, in every step of the cycle, the scalarized reward is in expectation at least $v$.

Secondly, we sum \eqref{eq:submart} over the cycle. The terms $h(x_t)$ telescope, and since $x_{\tau_{n-1}}$ and $x_{\tau_n}$ are both in $C$,
\begin{equation}\label{eq:cyclevalue}
    \mathbb{E}\big[\langle\lambda^{\star},G_n\rangle\mid\mathcal F_{\tau_{n-1}}\big] \geq v\,m_n+h(x_{\tau_{n-1}})-\mathbb{E}\big[h(x_{\tau_n})\mid\mathcal F_{\tau_{n-1}}\big] \geq v\,m_n-H_C .
\end{equation}
As the cycle has a random length, we first sum up to $(\tau_{n-1}+j)\wedge\tau_n$ and then let $j\to\infty$. This limit is valid because the partial sums are bounded by $2R_{\mathrm{max}}L_n+c_h\sum_{t=\tau_{n-1}}^{\tau_n}V(x_t)$, as $|v|\leq R_{\mathrm{max}}$, $|h|\leq c_hV$, and this bound has a finite conditional expectation. Thus, $V(x_{\tau_n})\leq\sup_CV$, and \eqref{eq:onestep} becomes $\mathbb{E}[V(x_{t+1})\mid\mathcal F_t]+\beta V(x_t)\leq V(x_t)+b$, which summed over the cycle gives $\beta\,\mathbb{E}[\sum_{t=\tau_{n-1}}^{\tau_n-1}V(x_t)\mid\mathcal F_{\tau_{n-1}}]\leq V(x_{\tau_{n-1}})+bM_1$ by monotone convergence and Lemma \ref{lem:cycles}.

Now we can decompose the remainder. By \eqref{eq:cyclevalue}, $-\mathbb{E}[\langle\lambda^{\star},G_n-L_np\rangle\mid\mathcal F_{\tau_{n-1}}]\leq(\langle\lambda^{\star},p\rangle-v)\,m_n+H_C$, and dividing by $m_n\geq L_{\mathrm{min}}$ as in \eqref{eq:remainder} gives $\varepsilon_n\leq[\langle\lambda^{\star},p\rangle-v]^++H_C/L_{\mathrm{min}}$. Lastly, we bound $\langle\lambda^{\star},p\rangle-v$. By Blackwell's condition \eqref{eq:blackwell} at $s=\bar{\mathbf r}_{\tau_{n-1}}$, $\langle\lambda^{\star},p\rangle\leq v^{\star}(\lambda^{\star})+\varepsilon_{\mathrm{approx}}$. Adding and subtracting $v^{\star}(\lambda)$ and $v(\theta,\lambda)$, and using $v^{\star}(\lambda)-v(\theta,\lambda)=\Delta_n$,
\begin{equation*}
    \begin{aligned}\langle\lambda^{\star},p\rangle-v(\theta,\lambda^{\star}) &\leq \varepsilon_{\mathrm{approx}}+\Delta_n+\big[v^{\star}(\lambda^{\star})-v^{\star}(\lambda)\big]+\big[v(\theta,\lambda)-v(\theta,\lambda^{\star})\big]\\ &\leq \varepsilon_{\mathrm{approx}}+\Delta_n+2R_{\mathrm{max}}\norm{\lambda-\lambda^{\star}} .\end{aligned}
\end{equation*}
The second inequality holds because $v(\theta,\cdot)$ and $v^{\star}$ are $R_{\mathrm{max}}$-Lipschitz: since $\norm{\bar{\mathbf r}_t}\leq R_{\mathrm{max}}$, the map $\lambda\mapsto\langle\lambda,\mathbb{E}^{x}_{\pi_\theta}\bar{\mathbf r}_t\rangle$ is $R_{\mathrm{max}}$-Lipschitz, and infima, limits inferior and suprema preserve this. Since the right-hand side is nonnegative, inserting it in the bound on $\varepsilon_n$ gives \eqref{eq:decomp}.
\end{proof}
The terms in \eqref{eq:decomp} characterize how a cycle can miss Blackwell's condition, such as the actor has not yet optimized ($\Delta_n$), the policy class cannot satisfy the condition exactly ($\varepsilon_{\mathrm{approx}}$, which includes the restriction imposed by the exploration parameter), the direction $\lambda$ lags the optimal $\lambda^\star$, and the cycle starts and ends at different points of $C$.

\begin{corollary}[Asymptotic guarantee]\label{cor:asymptotic}
Under the assumptions of Proposition \ref{prop:remainder}, almost surely and for every strategy of Player $2$,
\begin{equation*}
    \limsup_{t\to\infty} D(\bar{\mathbf r}_t, T)\leq\limsup_{n\to\infty}\Delta_n+\varepsilon_{\mathrm{approx}}+2R_{\mathrm{max}}\limsup_{n\to\infty}\norm{\lambda_n-\lambda^{\star}_n}+\frac{H_C}{L_{\mathrm{min}}} .
\end{equation*}
\end{corollary}

Corollary \ref{cor:asymptotic} follows by inserting \eqref{eq:decomp} in \eqref{eq:guarantee}. The first and third terms vanish when the actor and the direction converge, the second when the policy class satisfies Blackwell's condition exactly, and the last decreases with the minimum cycle length.

\section{Numerical results}\label{sec:numerical}

We consider a two-resistance, two-capacitance ($2R2C$) thermal network with a moisture node representing a $50\,\mathrm{m}^2$ room heated by an air-source heat pump and sampled every $15$ minutes. Player $1$ chooses among $10$ combinations of heating, ventilation, humidification, and dehumidification. Player $2$ shifts the ambient temperature and sets the occupancy adversarially. The objectives are $r_1=-$(violation of the comfort band $[20.5,24]\,^{\circ}\mathrm{C}$, in K), $r_2=-$(energy cost, EUR/h) and $r_3=-$(violation of the humidity band $[30,60]\,\%$), scaled by $(1,4,20)$, and $T=\{\mathbf r\geq\mathbf b\}$ with $\mathbf b=(-0.5,-0.20,-0.03)$. The weather, the tariff, and the equipment load are Markov processes, so the plant is time-homogeneous. The safety limits of the installation bound the reachable set. We consider a box $C$ that contains it and $L_{\mathrm{min}}=96$, so Assumption \ref{assump:drift} holds, and every cycle lasts one day. We compare against four baselines with identical seeds and interaction budgets. The four baselines are: 1. Fixed equal weights, 2. primal-dual CMDP maximizing comfort subject to the cost and humidity constraints, 3. \citet{mannor2004geometric} with $J=15$ learners, 4. rule-based thermostat. All learning methods share the actor-critic, features, and step sizes of Algorithm \ref{alg:mopg}, and are run for $180$ days. 
\begin{table}[t]
\centering
\small
\caption{Terminal distance to $T$ (mean $\pm$ standard error over $12$ seeds) and realized objectives after $180$ days. Target: comfort $\leq 0.5\,\mathrm{K}$, humidity $\leq 0.03$, cost $\leq 0.20\,\mathrm{EUR/h}$.}
\label{tab:main}
\begin{tabular}{lccccc}
\toprule
method & learners & $D(\bar{\mathbf r},T)$ & comfort [K] & humidity & cost [EUR/h]\\
\midrule
\textbf{proposed} & $1$ & $0.064\pm0.016$ & $0.524$ & $0.029$ & $0.207$\\
fixed, equal & $1$ & $0.116\pm0.021$ & $0.612$ & $0.020$ & $0.195$\\
Mannor--Shimkin & $15$ & $0.209\pm0.019$ & $0.658$ & $0.036$ & $0.209$\\
Lagrangian CMDP & $1$ & $0.316\pm0.097$ & $0.804$ & $0.012$ & $0.201$\\
thermostat & -- & $3.208\pm0.115$ & $0.003$ & $0.190$ & $0.115$\\
\bottomrule
\end{tabular}
\end{table}
\begin{figure}[t]
\centering
\includegraphics[width=0.86\linewidth]{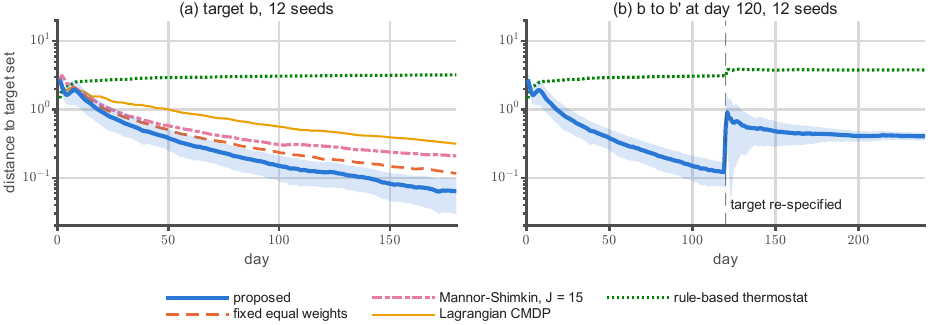}
\caption{Distance of the running average to the Target set, mean over $12$ seeds. The shaded band is the $95\,\%$ confidence interval of the proposed method. (a) Target $\mathbf b$. (b) The Target set is re-specified from $\mathbf b$ to $\mathbf b'$ at day $120$ and the running average restarts.}
\label{fig:results}
\end{figure}
The proposed method ends closest to $T$ (Table \ref{tab:main}, Figure \ref{fig:results}(a)). \citet{mannor2004geometric} maintains fifteen actor-critics, each of which receives only a share of the interaction budget, and ends three times farther from $T$ than the single learner. The thermostat keeps the temperature in the band but cannot respond to the adversary as expected. At day $120$ the operator re-specifies the Target set to $\{\mathbf r\geq\mathbf b'\}$ with $\mathbf b'=(-0.10,-0.23,-0.01)$, for a comfort-critical use with a larger budget, and the running averages restart (Figure \ref{fig:results}(b)). The Blackwell direction re-points at the new set (without any tuning), and the distance decreases again until it levels off, which indicates that the new set is at the limit of what the policy class attains against the adversary, i.e.\ the term $\varepsilon_{\mathrm{approx}}$ of Proposition \ref{prop:remainder}. The exploration parameter must be neither too small nor too large: for small $\zeta$ the policy becomes nearly deterministic before it has explored and most simulations fail to reach the target set, and for large $\zeta$ it stays too close to uniform to reach $T$.

\section{Conclusions}\label{sec:conclusions}

We have extended approachability-based multiobjective RL to Borel state spaces with a single two-timescale actor-critic whose distance to the Target set is bounded by the Blackwell remainder. Although we considered continuous action spaces, the results presented here are limited to the considered policy class, and proposition \ref{prop:remainder} quantifies the error due to our choice. Future work will involve extension from finite action spaces to Borel action spaces and thus refining Theorem \ref{thm:main} and Proposition \ref{prop:remainder} by getting rid of $\varepsilon_{\mathrm{approx}}$. Another direction is to consider the general-sum case, where the natural analogue is the correlated-equilibrium construction of \citet{hart2000simple}, which will be explored in this setting.

\acks{This work was supported by Poul Due Jensen Fonden under the project SWiM.}

\bibliography{mybib,addrefs}

\end{document}